\documentclass[11pt,dvipsnames]{extarticle}

\newif\ifanonymous
\anonymousfalse

\usepackage{amsfonts}
\usepackage{amsmath}
\usepackage{amssymb}
\usepackage{amscd}
\usepackage{amsthm}
\usepackage{mathrsfs}
\usepackage{graphicx}
\usepackage{wasysym}
\usepackage{xcolor}
\usepackage{geometry}
\usepackage{physics}
\usepackage[12hr]{datetime}
\usepackage{textcomp}
\usepackage{hyperref}
\usepackage[utf8]{inputenc}
\usepackage[T1]{fontenc}
\usepackage[nameinlink, noabbrev, capitalize]{cleveref}
\crefname{algocf}{Algorithm}{Algorithms}
\Crefname{algocf}{Algorithm}{Algorithms}
\usepackage{appendix}
\usepackage{mathtools}
\usepackage{xfrac}
\usepackage{nicefrac}
\usepackage{url}
\usepackage{environ}
\usepackage{nicematrix}
\usepackage[hang,flushmargin]{footmisc}
\usepackage{algorithm2e}
\usepackage[shortlabels]{enumitem}
\usepackage{multicol}
\usepackage{tabularx}
\usepackage{xpatch}
\usepackage{changepage}

\ifanonymous
\else
\usepackage[sc,osf]{mathpazo}
\usepackage[euler-digits,small]{eulervm}
\AtBeginDocument{\renewcommand{\hbar}{\hslash}}
\fi

\usepackage[backend=biber,style=alphabetic,backref=true,sorting=anyt,maxalphanames=3,minalphanames=2,maxbibnames=99,doi=false,isbn=false,url=false,eprint=false]{biblatex}

\ifanonymous

\else

\renewbibmacro{in:}{}
\DeclareFieldFormat[article,inbook,incollection,inproceedings,patent,thesis,unpublished]{title}{#1}
\DefineBibliographyStrings{english}{backrefpage={},backrefpages={}}

\DeclareFieldFormat{backrefparens}{\addperiod#1}
\xpatchbibmacro{pageref}{parens}{backrefparens}{}{}
\fi

\hypersetup{colorlinks=true, urlcolor=NavyBlue, linkcolor=Mahogany, citecolor=ForestGreen, pdfborder={0,0,0}, linktoc=all}

\let\oldFootnote\footnote
\newcommand\nextToken\relax
\renewcommand\footnote[1]{\oldFootnote{#1}\futurelet\nextToken\isFootnote}
\newcommand\isFootnote{\ifx\footnote\nextToken\textsuperscript{,}\fi}

\newcommand{\mb}{\mathbb}

\newcommand{\mc}{\mathcal}

\newcommand{\n}{\enspace}
\newcommand{\tx}{\text}
\newcommand{\ol}{\overline}

\newcommand{\vst}{\vspace{0.2cm}}

\newcommand{\poly}{\tx{poly}}
\newcommand{\ord}{\tx{\normalfont ord}}

\newcommand{\email}[1]{\href{mailto:#1}{\textcolor{NavyBlue}{\texttt{#1}}}}
\newcommand{\F}{\mathbb{F}}
\newcommand{\eps}{\epsilon}

\newcommand{\PGL}{\mathrm{PGL}}
\renewcommand{\Tr}{\mathrm{Tr}}
\newcommand{\id}{\mathrm{id}}

\mathchardef\mhyphen="2D

\newcommand{\SD}{\mathrm{SD}}
\newcommand{\ShamirSS}{\mathrm{ShamirSS}}
\newcommand{\Bad}{\mathrm{Bad}}
\newcommand{\Poles}{\mathrm{Poles}}
\newcommand{\ShortPeriod}{\mathrm{ShortPeriod}}
\newcommand{\Share}{\mathrm{Share}}
\newcommand{\Img}{\mathrm{Im}}
\newcommand{\Pp}{\mb{P}^1}

\usepackage{thmtools}
\usepackage{thm-restate}

\theoremstyle{theorem}
\newtheorem{theorem}{Theorem}[section]
\newtheorem{fact}[theorem]{Fact}

\newtheorem*{claim*}{Claim}
\newtheorem{proposition}[theorem]{Proposition}

\newtheorem{lemma}[theorem]{Lemma}
\newtheorem{corollary}[theorem]{Corollary}

\newtheorem{openquestion}[theorem]{Open Question}

\newtheoremstyle{TheoremNum}
{\topsep}{\topsep}{\itshape}{}{\bfseries}{.}{ }{\thmname{#1}\thmnote{ \bfseries #3}}
\theoremstyle{TheoremNum}

\theoremstyle{definition}
\newtheorem{definition}[theorem]{Definition}

\newtheorem{remark}[theorem]{Remark}

\makeatletter
\NewEnviron{restatethm}[2]{%
	\protected@write\@auxout{}{\string\@restatetheorem{#1}{\detokenize\expandafter{\BODY}}}%
	\begin{theorem}[#2]\label{#1}\BODY\end{theorem}%
}
\newcommand{\@restatetheorem}[2]{\expandafter\gdef\csname restatethm@#1\endcsname{#2}}
\newcommand{\restatethmnow}[2]{%
	\begingroup\renewcommand{\thetheorem}{\ref{#1}}%
	\begin{theorem}[#2]\csname restatethm@#1\endcsname\end{theorem}\endgroup}
\makeatother



\makeatletter
\usepackage{titletoc}

\titlecontents*{subsection}
[1.5em]
{\small\raggedright} % \raggedright allows clean line drops and prevents margin overshoots
{\hbox\bgroup\color{\@linkcolor}\thecontentslabel~} % Starts an unbreakable box and inherits your exact hyperlink color
{\hbox\bgroup} % Starts an unbreakable box (for unnumbered sections)
{\egroup} % Closes the unbreakable box (effectively replacing the page number)
[\textcolor{\@linkcolor}{;}\quad] % The separator, which acts as the ONLY allowed place for a line break
\makeatother

\usepackage{titling}
\title{Partial Derandomization for Leakage-Resilient Shamir's Secret Sharing over Composite Order Fields\thanks{A preliminary version of this work is due to appear at ITC 2026.}}
\author{S.~Venkitesh\thanks{Institute for the Theory of Computing, The Stein Faculty of Computer and Information Science, Ben Gurion University of the Negev, Beersheva, Israel. Email: \email{venkitesh.mail@gmail.com}}}
\date{}

\begin{document}

\maketitle

\vspace{-1cm}
\begin{abstract}
	We make progress on the question of constructing explicit evaluation places for leakage-resilient Shamir's secret sharing, over composite order fields.  Previously, Maji et al. (EUROCRYPT 2024) showed that random evaluation places yield Shamir's secret sharing over the composite order field \(\F_{p^d}\) that is statistically secure against physical-bit leakage.  Later, Nguyen (EUROCRYPT 2025) established a \emph{dichotomy} that linear code-based secret-sharing scheme over the field \(\F_{p^d}\) is either statistically secure or completely insecure against such leakage.
	
	Building upon Nguyen's dichotomy, we present a partial derandomization of evaluation places, improving upon the Maji et al. result for a restricted regime of parameters.  We replace the random choice of \(n\) independent evaluation places by the iterates \(x_j = \Phi^j(x_0)\) of a simple fixed rational function \(\Phi\), where the initial point \(x_0 \in \F_{p^d}^*\) is randomly chosen.  The randomness in the evaluation places thus drops from \(nd \log p\) bits to \(d\log p\) bits.  Our construction is valid for the regime \(n = O(d/\log_p d)\), and any reconstruction threshold \(k \ge 2\); in fact, the scheme attains perfect security (statistical distance exactly zero) against single-block leakage. Our technique is a partial fraction nondegeneracy argument that exploits the distinct poles of the rational iterates.
\end{abstract}

%\cleardoublepage
{\small\tableofcontents}
\cleardoublepage

%% ============================================================
%%  SECTION 1: INTRODUCTION
%% ============================================================
\section{Introduction}\label{sec:intro}

Threshold secret-sharing schemes, such as Shamir's scheme~\cite{shamir-1979-secret}, distribute a secret among \(n\) parties so that any \(k\) of them can reconstruct the secret, while any fewer than \(k\) parties learn nothing about it.  In the standard corruption model, an adversary obtains the complete shares of some parties and has no information about the remaining shares.

Side-channel attacks have repeatedly circumvented this all-or-nothing assumption.  Rather than corrupting entire shares, a side-channel adversary accumulates small amounts of information, like individual bits, power traces, and timing signals, from \emph{all} shares simultaneously (see Ishai, Sahai, and Wagner~\cite{ISW03}).  The mathematical abstraction of such threats is \emph{independent local leakage}:  the adversary applies a bounded function to each share independently and observes only the outputs.  The study of locally leakage-resilient secret sharing was initiated by Benhamouda, Degwekar, Ishai, and Rabin~\cite{BDIR21} and is also implicit in the work of Goyal and Kumar~\cite{GK18}.

A particularly natural leakage model, introduced by~\cite{ISW03}, probes \emph{physical bits} in the memory storing each share.  Since elements of a finite field \(\F_{p^d}\) are stored as \(d\) coordinates over \(\F_p\), each represented as a \(\lceil\log_2 p\rceil\)-bit binary string, a physical bit probe extracts a single bit from this representation.  Additive secret sharing is known to be vulnerable to such probes:  the parity-of-parities attack of Maji, Nguyen, Paskin-Cherniavsky, Suad, and Wang~\cite{MNPSW21} leaks the least significant bit of each share and distinguishes secrets with advantage \((2/\pi)^n\) over any prime field \(\F_p\).  Over characteristic-\(2\) fields, this attack can distinguish secrets \(0\) and \(1\) with certainty.  Shamir's secret sharing inherits these vulnerabilities if its evaluation places are chosen carelessly~(\cite{MNPSW21}, and Costes and Stam~\cite{CS21}).

The question, then, is whether Shamir's scheme can be instantiated to resist physical bit probes.  Over prime fields,~\cite{MNPSW21} showed that choosing evaluation places uniformly at random yields a leakage-resilient scheme with high probability.  Maji, Nguyen, Paskin-Cherniavsky, and Ye~\cite{MNPY24} extended this randomized construction to composite order fields, particularly fields of characteristic \(2\), and also gave an exact security classifier for \(k=2\) against single-block leakage.  However, no explicit evaluation places were known for \(k > 2\) or for the general block-leakage regime over composite fields.  Even for \(k = 2\), the classifier of~\cite{MNPY24} certifies individual evaluation-place tuples but does not produce a closed form family.  Separately, Hwang, Maji, Nguyen, and Ye~\cite{HMNY24} initiated the study of explicit constructions over prime fields, providing classifiers and derandomized evaluation places for Mersenne and Fermat primes in the full-threshold regime \(n = k\).

In this work, we present a partial derandomization of~\cite{MNPY24}.  Building on the dichotomy of Nguyen~\cite{nguyen-2025}, which establishes that any linear code-based secret-sharing scheme over the field \(\F_{p^d}\) is either \emph{perfectly} secure or completely insecure against physical-bit leakage, we replace the random choice of \(n\) independent evaluation places by a structured one-parameter family:  the evaluation places are \(x_j = \Phi^j(x_0)\), where \(\Phi(x) \coloneqq \alpha x/(x+1)\),  \(\alpha\) is a fixed multiplicative generator of \(\F_{p^d}^*\), and \(x_0\in \F_{p^d}^*\) is a randomly chosen \emph{base point}.  As it turns out, \(\Phi\) is a \emph{M\"obius transformation}, and the utility of such structured algebraic transformations in defining evaluation places seems unexplored so far.  The randomness needed to specify the \(n\) evaluation places drops from \(nd \log p\) bits to \(d\log p\) bits, required to choose the base point \(x_0\).  The above mentioned dichotomy argument reduces the security question to a \emph{sufficient} full-rank condition on a test matrix, which we verify for our chosen points by observing a nondegeneracy property that exploits the distinct poles of the iterates.

The dichotomy itself is implicit, for the case \(k = 2\), in the linear-algebraic arguments of~\cite{MNPY24}, and~\cite{nguyen-2025} later extends this to the case of general \(k\), as well as to general linear code-based secret sharing.  We also give an alternative direct derivation of the dichotomy in our setting for completeness, in \Cref{app:dichotomy-proofs}.

\subsection{Motivation}\label{subsec:motivation}

The problem of \emph{derandomization of evaluation places} is the central motivation for our work.

The randomized constructions of~\cite{MNPSW21,MNPY24} demonstrate that most evaluation places are secure.  In practice, however, they require trusted public randomness, for instance a randomness beacon, to select the places.  An adversary who can influence the random seed may steer the construction toward vulnerable evaluation places, unbeknownst to the honest parties.  This concern is not merely theoretical; the NIST standardization effort for threshold cryptographic schemes~(see Brand\~{a}o and Peralta~\cite{NIST-threshold}), and the practical deployment considerations studied by Faust, Masure, Micheli, Orlt, and Standaert~\cite{FMMOS24} both underscore the need for deterministic, verifiable instantiations.

Ideally, one would like evaluation places that are fixed by the field specification alone, with a proof that they resist physical bit probes.  Towards this,~\cite{HMNY24} initiated the study of classifier and derandomization constructions over prime fields.  Their techniques, however, are inherently tied to the structure of prime fields:  they exploit square wave orthogonality and \(2\)-adic valuations of rational approximations, both of which rely on the nonlinearity of the bit-extraction map over \(\F_p\).  Over composite fields \(\F_{p^d}\) with \(d \ge 2\), the analogous coordinate-extraction map is \(\F_p\)-linear, and the Fourier-analytic machinery of~\cite{HMNY24} does not apply.  The present work exploits this linearity to make partial progress on the derandomization problem in the composite-field setting; the precise results are stated in~\Cref{subsec:results}.

\subsection{Our results}\label{subsec:results}

We work over \(\F = \F_{p^d}\), where \(p\) is any prime and \(d \ge 2\).  Fix a primitive polynomial \(\Pi(t) \in \F_p[t]\) of degree \(d\), and let \(\alpha\) be a root of \(\Pi\) in \(\F\), so that \(\ord(\alpha) = p^d - 1\) and \(\F = \F_p[\alpha]\) with canonical basis \(\mc{B} = \{1, \alpha, \alpha^2, \ldots, \alpha^{d-1}\}\).

Our construction is based on a rational function.  Define the \emph{step operator}
\[
\Phi(x) = \frac{\alpha x}{x + 1},
\]
which is a M\"obius transformation on the projective line \(\Pp(\F) := \F \cup \{\infty\}\) with projective order \(p^d - 1\).  Our evaluation places are the iterates
\[
x_j = \Phi^j(x_0), \quad j = 0, 1, \ldots, n-1,
\]
where \(\Phi^j\) denotes the \(j\)-th iterate of \(\Phi\) under composition, applied to a base point \(x_0 \in \F^*\).

The choice of~\(\Phi\) is not arbitrary.  The single structural property of the orbit that drives the entire security analysis is that the iterates~\(\Phi^0, \Phi^1, \ldots, \Phi^{n-1}\), viewed as M\"obius transformations on the projective line \(\Pp(\F)\), have \emph{pairwise distinct poles}:  the pole of \(\Phi^0 = \id\) is \(\infty\), and the pole of \(\Phi^j\) for \(j \ge 1\) is a point of \(\F^*\), these being distinct for distinct \(j\).  This pole-distinctness is what powers the partial fraction nondegeneracy argument in \Cref{sec:single-block}.  In fact, motivated from the construction of \emph{folded Reed-Solomon codes}~\cite{guruswami-rudra-08}, which are polynomial codes that can be \emph{list decoded} up to the information theoretic limit, a more obvious and simpler candidate is \(\Psi(x) = \alpha x\).  Here, every iterate \(\Psi^j(x) = \alpha^j x\) is a pure dilation with its only pole at~\(\infty\).  So the iterates share a single common pole and the partial fraction argument collapses entirely.  We make this contrast precise in~\Cref{rem:dilation-fails}.  The role of the~\(+1\) in the denominator of~\(\Phi\) is therefore not cosmetic; it moves the pole of every non-identity iterate to a distinct point of \(\F^*\), placing the orbit in sufficiently generic position so that a good bad-set bound becomes possible.

Our main results are as follows.
\begin{restatethm}{thm:main-single-block}{Perfect security against single-block leakage}
Let \(p\) be any prime, \(d \ge 2\), \(k \ge 2\), and \(k\le n \le d(k-1)\) with \(n \le p^d - 1\).  Let \(\alpha\) be a multiplicative generator of \(\F_{p^d}^*\) and \(\Phi(x) = \alpha x/(x+1)\).  There exists a set \(\Bad \subset \F_{p^d}^*\) with
\[
|\Bad| \le n + n(dp)^n
\]
such that for any \(x_0 \in \F_{p^d}^* \setminus \Bad\), the evaluation places \(\vec{X}=(x_0,\ldots,x_{n-1})\) give a scheme\\ \(\ShamirSS(n, k, \vec{X})_{\F_{p^d}}\) that is \emph{perfectly secure} against all single-\(\F_p\)-block leakage patterns.

In particular, the scheme is perfectly secure against all single-physical-bit-per-share leakage (for any \(p\)), and against any leakage that applies an arbitrary function \(g_j : \F_p \to S_j\) to a single \(\F_p\)-block of each share, where \(S_0, \ldots, S_{n-1}\) are arbitrary nonempty finite sets.

A good base point \(x_0\) exists whenever \(d > n(1 + \log_p d) + \log_p(2n)\), giving a practical range of \(n = O\big(d/\log_p(pd)\big)\) parties.
\end{restatethm}
\noindent \Cref{thm:main-single-block} is proved in \Cref{sec:single-block}.  The bad set and perfect security are \Cref{thm:single-block}, the existence claim is \Cref{cor:existence-single}, and the two particular claims are \Cref{pro:sub-block} and \Cref{cor:physical-bit}.

The security guarantee is \emph{perfect}:  the statistical distance is zero, not merely exponentially small.  Within the parameter range \(n = O(d/\log_p d)\), this is a qualitative strengthening of the statistical security \(\eps = 2^{-\Omega(d)}\) achieved by the randomized constructions of~\cite{MNPY24,nguyen-2025}.  The trade-off is on party count:~\cite{MNPY24,nguyen-2025} support \(n = O(dk/\log_p d)\) parties versus our \(n = O(d/\log_p d)\), reflecting the cost of restricting the evaluation places to a one-parameter family.

The construction is a partial derandomization.  Once a base point \(x_0\) is fixed, the evaluation places \(x_j = \Phi^j(x_0)\) are determined by the field representation alone; the randomness needed to specify the \(n\) places thus drops from \(nd \log p\) bits (for \(n\) independent random places, as in~\cite{MNPY24}) to \(d\log p\) bits needed to fix the single \(x_0 \in \F^*\).  In the headline regime, \(x_0 = \alpha\) is a structured candidate; whether it is good depends on the choice of primitive polynomial \(\Pi\) and is certified by~\Cref{alg:classifier} (\Cref{subsec:classifier}), a sound test whose \textsc{good} verdict guarantees perfect security, with the same worst-case runtime as the classifying algorithm of~\cite{nguyen-2025}.

For multi-block leakage, we prove the following.

\begin{restatethm}{thm:main-multi-block}{Multi-block leakage, improved bound}
Under the hypotheses of \Cref{thm:main-single-block}, fix an admissible multi-block leakage pattern with \(M < d - \log_p(2n)\) total \(\F_p\)-blocks leaked. Then the bad set has size \(\le n + n \cdot p^M\), and a good \(x_0\) always exists.

For universality over all admissible multi-block patterns with \(\le M\) blocks total, we get the bound \(|\Bad| \le n + n(dpe)^M\), and a good \(x_0\) exists when \(M = O(d/\log_p d)\).
\end{restatethm}
\noindent For the proof, we observe that the Vandermonde has nontrivial kernel, and the direct bound from \Cref{thm:multi-block-fixed-pattern} is used; obtaining a polynomial bound for this regime remains open (\Cref{subsec:discussion}).

\subsection{Comparison with prior work}\label{subsec:comparison}

\Cref{tab:comparison} summarizes the comparison with~\cite{MNPY24,nguyen-2025}, and a few gaps remain.  First, there is a factor-\(k\) loss in the maximum number of parties for large \(k\); we conjecture this loss is inherent to single-parameter constructions, but leave the question open (\Cref{subsec:discussion}).  Second, multi-block universality is limited to \(M = O(d/\log_p d)\) by the pattern enumeration bottleneck.  Closing these gaps, particularly by obtaining a polynomial bad-set bound, is an interesting direction for future work; see \Cref{subsec:discussion}.

\begin{table}[ht]
	\centering
	\renewcommand{\arraystretch}{1.3}
	\begin{tabularx}{\textwidth}{@{} >{\raggedright\arraybackslash}p{4.2cm} *{3}{>{\centering\arraybackslash}X} @{}}
		\hline
		\textbf{Aspect} & \cite{MNPY24} & \cite{nguyen-2025} & \textbf{This work} \\
		\hline
		Evaluation places & Random in \((\F^*)^n\) & Fixed; arbitrary & \(\Phi^j(x_0)\) orbit \\
		Multipliers \(v_i\) (in shares \(v_i P(X_i)\)) & All \(1\) & Random in \((\F^*)^n\) & All \(1\) \\
		Security guarantee & \(\eps = 2^{-\Omega(d)}\) & \(\eps = 2^{-\Omega(d)}\) & \(\SD = 0\) for \(x_0 \notin \Bad\) \\
		Max parties (general \(k\)) & \(O(dk/\log_p d)\) & \(O(dk/\log_p d)\) & \(O(d/\log_p d)\) \\
		Multi-block, \(M < d\) & \(\eps = 2^{-\Omega(d)}\) & \(\eps = 2^{-\Omega(d)}\) & Perfect (fixed pattern) \\
		Random bits in evaluation places & \(nd \log p\) & \(0\) & \(d \log p\) (for \(x_0\)) \\
		Random bits in multipliers & \(0\) & \(nd \log p\) & \(0\) \\
		\hline
	\end{tabularx}
	\caption{Comparison of our results with the randomized results of~\cite{MNPY24,nguyen-2025}.  Here \(2\le k\le n\).}\label{tab:comparison}
\end{table}

\subsection{Related work}\label{subsec:related-work}

We survey the most relevant prior results.

\paragraph*{Leakage-resilient secret sharing.}  \cite{BDIR18,BDIR21} proved that Shamir's scheme is leakage-resilient against arbitrary single-bit local leakage when \(k/n\) exceeds a constant threshold.  Subsequent works~\cite{MNPW22,KK23,Nguyen24} improved this threshold, with the current best being \(k \ge 0.69n\).  These results hold for \emph{all} evaluation places but require a large reconstruction-to-party ratio.  Among other lines of work that bring linear algebraic structure to bear on these problems, Koga and Abe~\cite{koga-abe-2026} determine tight bounds on the local leakage resilience of the additive \((n,n)\)-threshold scheme via the eigenvalues of circulant matrices.

The complementary regime (small \(k\) relative to \(n\)) is where the choice of evaluation places becomes critical.  \cite{MNPSW21} initiated the study of this regime over prime fields, and~\cite{MNPY24} extended it to composite fields.  Other constructions of leakage-resilient (non-Shamir) secret sharing, including those based on algebraic geometry codes, appear in~\cite{KMS19,MPSW21} and the references therein.

\paragraph*{Physical bit probing and the {\normalfont\cite{MNPY24}} framework.}  The closest prior randomized construction we build upon is~\cite{MNPY24}.  Their construction works over \(\F_{p^d}\) for any prime \(p\) and combines three ingredients:  \n(i) a Fourier-analytic upper bound on statistical distance in terms of dual GRS codewords and Fourier coefficients of block-indicator functions, \n(ii) a B\'ezout-type bound on simultaneous zeros of polynomial systems over composite fields~\cite{zhao-2012-exponential-sums,bafna-sudan-velusamy-xiang-2021-isolated}, and \n(iii) a generalized Vandermonde determinant analysis.

They also proved an exact security classifier for \(k = 2\).  Both the dichotomy at \(k = 2\) and its proof technique (linear algebra plus the cosets-of-subspaces structure) are the entry point for our test-matrix derivation.

\paragraph*{The dichotomy result of {\normalfont\cite{nguyen-2025}}.}  A perfect-or-completely-insecure dichotomy for any linear code-based secret-sharing scheme was given by~\cite{nguyen-2025}, over a binary extension field against arbitrary physical-bit leakage, along with a complete characterization of the insecure leakages via minimal codewords of the dual of the binary image.  A remark therein also explicitly extends both results to composite order fields \(\F_{p^d}\) and to subfield-coordinate leakage, which is the present setting.  The work also gives a classifying algorithm whose worst-case cost matches our verification cost (\Cref{pro:classifier-complexity}).  Beyond the dichotomy,~\cite{nguyen-2025} gives a Monte-Carlo construction of GRS-based linear code-based secret sharing with random multipliers and fixed evaluation places, achieving \(M \le (k-1)\lambda/\poly(\log \lambda)\) bits of physical-bit leakage tolerance with overwhelming probability over the multiplier choice.

Our contribution is complementary to~\cite{nguyen-2025} on the construction side.  We give a construction of evaluation places \(x_j = \Phi^j(x_0)\) using a suitable rational function \(\Phi\), that fixes both the multipliers (all equal to \(1\), as in standard Shamir's scheme) and the evaluation places to a one-parameter family, with the security question reducing (via the dichotomy of~\cite{nguyen-2025}) to a closed form bad-set bound on the base point \(x_0\).  The trade-off is on party count:~\cite{nguyen-2025} supports \(n = O(dk/\log_p d)\) parties; we support \(n = O(d/\log_p d)\), reflecting the cost of restricting to a single-parameter family.

\paragraph*{Explicit constructions over prime fields.}  \cite{HMNY24} constructed efficient classifiers for evaluation places over Mersenne and Fermat primes in the \(n = k\) regime, connecting leakage resilience to orthogonality of square wave functions and \(2\)-adic valuations of rational approximations.  They provide explicit secure evaluation places for \(n = k = 2\) and lift to \(n = k > 2\) via a Fourier-analytic theorem.  Their techniques are specific to prime fields, where bit extraction is a nonlinear map.  Over composite fields, the analogous map is \(\F_p\)-linear, enabling a fundamentally different (and in several respects simpler) approach.

\paragraph*{Acknowledgment.}  The author thanks the anonymous reviewers of ITC 2026 for critical feedback, as well as for pointing to~\cite{nguyen-2025} and the dichotomy result therein.

\subsection{Technical overview}\label{subsec:technical-overview}

The proof proceeds in three phases.  We outline each in turn.

\paragraph*{Phase~1: The perfect dichotomy (\Cref{sec:prelim}).}  The key structural observation is that the coordinate extraction map \((\cdot)_i : \F \to \F_p\), defined by \((x)_i = \Tr(x \cdot f_i^*)\) where \(\{f_i^*\}\) is the dual basis with respect to the field trace, is \(\F_p\)-linear.  It follows that the leakage map \(L_{\vec{i}} : \F^{k-1} \to \F_p^n\) is \(\F_p\)-linear.

This has a clean consequence.  By the structure of cosets of \(\F_p\)-subspaces, the leakage distribution for secret \(s\) is uniform over a coset of \(\Img(L_{\vec{i}})\).  Two cosets are either identical or disjoint, yielding the dichotomy \(\SD \in \{0, 1\}\).  Consequently, surjectivity of \(L_{\vec{i}}\), equivalently a test matrix \(\Theta_{\vec{i}}\) having full column rank \(n\) over \(\F_p\), is a \emph{sufficient} condition for perfect security against the pattern \(\vec{i}\), since it forces every coset to equal \(\F_p^n\).  (Perfect security against \(\vec{i}\) is in fact the weaker condition that the secret's coordinate vector \(((s)_{i_j})_{j}\) lies in \(\Img(L_{\vec{i}})\) for every \(s\), which is the minimal codeword characterization of~\cite[Theorem~6]{nguyen-2025}; see also~\cite[Theorem~5]{MNPY24} for the \(k=2\) case.  We use only the sufficient direction.)  An alternative derivation of the dichotomy \(\SD\in\{0,1\}\) in our setting is given in~\Cref{app:dichotomy-proofs}.  The rank criterion is the working tool for the rest of the proof.

\paragraph*{Phase~2: partial fraction nondegeneracy (\Cref{sec:single-block}).}  The test matrix \(\Theta_{\vec{i}}\) encodes the map \(\vec{c} \mapsto (\sum_j c_j \eta^{(i_j)} x_j^\ell)_{\ell=1}^{k-1}\).  To show full rank, we must prove that no nontrivial \(\F_p\)-linear combination \(G_\ell(x) = \sum_j c_j \eta^{(i_j)} (\Phi^j(x))^\ell\) vanishes identically.

The argument is a partial fraction analysis.  Since the M\"obius transforms~\(\Phi^0, \Phi^1, \ldots, \Phi^{n-1}\) have pairwise distinct poles (\Cref{cor:distinct-poles}), the residue of~\(G_\ell\) at each pole is determined by a single summand.  This pole-distinctness is the unique structural property of the orbit on which the entire argument hinges; \Cref{rem:dilation-fails} shows that the naive alternative~\(\Psi(x) = \alpha x\) fails exactly here.  If \(G_\ell \equiv 0\), all coefficients must vanish.  The resulting zero count \(\le n\ell\) at the smallest power \(\ell = 1\) gives \(\le n\) bad base points per coefficient vector, and a union bound over all patterns and coefficients yields the bad set of \Cref{thm:main-single-block}.

\paragraph*{Phase~3: Multi-block extension (\Cref{sec:multi-block}).}  For multi-block leakage, the kernel condition involves \(\F\)-valued coefficients \(d_j = \sum_r c_{j,r} \eta^{(i_j^{(r)})}\) arising from within-share linear combinations.  The admissibility of the leakage pattern ensures that \(d_j \ne 0\) whenever the corresponding \(c_{j,r}\) are not all zero.  Once this is established, the partial fraction argument extends directly to \(\F\)-coefficients.

\subsection{Discussion and open problems}\label{subsec:discussion}

It is worth recording what our proofs will actually rely on regarding the step operator \(\Phi\).
\begin{enumerate}[(1)]
	\item \emph{Rational function structure:}  each \(\Phi^j\) is a degree-\(1\) rational function, which will enable the partial fraction analysis in \Cref{sec:single-block}.
	\item \emph{Distinct poles:}  \(\Phi^0, \ldots, \Phi^{n-1}\) have pairwise distinct poles in \(\Pp(\F)\), which is the hypothesis driving our nondegeneracy lemma.
\end{enumerate}
%\n(3) \emph{Explicit pairwise differences:}  the factor \((x_0 - (\alpha - 1))\) appears in every difference \(\Phi^{j'}(x_0) - \Phi^j(x_0)\), making the Vandermonde determinant nonzero at all non-excluded base points.
%This eliminates the union bound for \(n \le k - 1\) (\Cref{sec:function-field}).

%Our proofs make no appeal to the Ramanujan spectral gap, the expander mixing lemma, character sum bounds, or any equidistribution property.
%These tools are available for torus-based constructions (M\"obius maps on \(\mu_{q+1}\)) and could potentially be brought to bear on the open problems below.
%
%\paragraph*{Beyond Shamir.}  The linear-algebraic perspective taken here -- reducing perfect security to a rank condition on a test matrix derived from the leakage map \(L_{\vec{i}}\) -- depends only on the \(\F_p\)-linearity of coordinate extraction and the \(\F_p\)-linearity of the share map.  It should therefore extend to a broader class of algebraic secret-sharing schemes over composite order fields, including monotone span programs, where the analogous rank criterion would govern leakage resilience.  We leave a systematic development of this direction to future work.

Several natural questions remain.

\begin{openquestion}[Universal multi-block for \(2\le k\le n\) and \(M = \omega(d/\log_p d)\)]
	The pattern enumeration bottleneck limits universality to \(M = O(d/\log_p d)\).  Can this be overcome?
\end{openquestion}

\begin{openquestion}[Polynomial-size bad set]
	For a fixed admissible pattern, the Vandermonde matrix of evaluation places has nontrivial kernel, and the bad-set bound \(\le n \cdot p^M\) from \Cref{thm:multi-block-fixed-pattern} is exponential in \(M\). Can this be improved to a polynomial bound?
\end{openquestion}
%
%\begin{openquestion}[Statistical security for large leakage]
%Can the Ramanujan spectral gap on algebraic tori be used to achieve statistical security \(\eps = 2^{-\Omega(d)}\) for an explicit single-parameter construction, matching the tolerance of the randomized constructions of~\cite{MNPY24,nguyen-2025}?
%\end{openquestion}

\begin{openquestion}[Factor-\(k\) gap]
	The randomized constructions of~\cite{MNPY24,nguyen-2025} handle \(n = O(dk/\log_p d)\) parties; our single-parameter construction handles \(n = O(d/\log_p d)\).
	Can this be overcome?
\end{openquestion}

%\paragraph*{Phase~4: Vandermonde invertibility (\Cref{sec:function-field}).}  The multi-block bad-set bound from Phase~3 is exponential in \(M\).  The key identity is that
%\[
%\Phi^{j'}(x_0) - \Phi^j(x_0) = \frac{x_0(\alpha^j - \alpha^{j'})(x_0 - (\alpha - 1))}{(\alpha - 1)(C_{j'} x_0 + 1)(C_j x_0 + 1)}.
%\]
%The factor \((x_0 - (\alpha - 1))\) appears in \emph{every} pairwise difference.  Consequently, the Vandermonde determinant of \(\Phi^0(x_0), \ldots, \Phi^{n-1}(x_0)\) is nonzero for all \(x_0 \in \F^* \setminus (\{\alpha - 1\} \cup \Poles(n))\), which is exactly the structural exclusion set.

%% ============================================================
%%  SECTION 2: PRELIMINARIES
%% ============================================================
\section{Preliminaries}\label{sec:prelim}

Throughout, \(p\) denotes a prime and \(d \ge 2\) an integer.  We write \(\F = \F_{p^d}\) and \(\F^* = \F \setminus \{0\}\).  The security parameter is \(\lambda = d \log_2 p\).

\subsection{The field and coordinates}\label{subsec:field}

Fix a primitive polynomial \(\Pi(t) \in \F_p[t]\) of degree \(d\).  Then \(\F = \F_p[t]/\Pi(t)\), and a root \(\alpha\) of \(\Pi\) satisfies \(\ord(\alpha) = p^d - 1\), i.e., \(\alpha\) generates \(\F^*\).  The canonical basis of \(\F\) over \(\F_p\) is \(\mc{B} = \{1, \alpha, \alpha^2, \ldots, \alpha^{d-1}\}\).  Every \(x \in \F\) is uniquely written as \(x = \sum_{i=0}^{d-1} x_i \alpha^i\) with \(x_i \in \F_p\), and we define the \(i\)-th \(\F_p\)-coordinate of \(x\) by \((x)_i := x_i\).

\begin{definition}[Dual basis and shifting factors]\label{def:dual-basis}
Let \(\Tr := \Tr_{\F/\F_p}\) denote the field trace.  The \emph{dual basis} of \(\mc{B}\) with respect to \(\Tr\) is the unique basis \(\{f_0^*, f_1^*, \ldots, f_{d-1}^*\}\) of \(\F\) over \(\F_p\) satisfying \(\Tr(\alpha^i \cdot f_j^*) = \delta_{ij}\) for all \(i, j\).  The \emph{shifting factors} are \(\eta^{(i)} := f_i^* / f_0^*\) for \(i = 0, 1, \ldots, d-1\).
\end{definition}

In particular, \(\eta^{(0)} = 1\).  The dual basis exists and is unique because the Gram matrix \(G_{ij} = \Tr(\alpha^i \cdot \alpha^j)\) is invertible (by nondegeneracy of the trace form).
\begin{fact}[Coordinates via the trace]\label{fac:coord-trace}
	For all \(x \in \F\) and all \(i \in \{0, \ldots, d-1\}\) we have \((x)_i = \Tr(x \cdot f_i^*)\).  In particular each coordinate map \((\cdot)_i : \F \to \F_p\) is \(\F_p\)-linear.
\end{fact}
\begin{proof}
	Write \(x = \sum_{m=0}^{d-1} (x)_m \alpha^m\) with \((x)_m \in \F_p\).  Since \(\Tr\) is \(\F_p\)-linear and \(\Tr(\alpha^m \cdot f_i^*) = \delta_{mi}\) (\Cref{def:dual-basis}), we get \(\Tr(x \cdot f_i^*) = \sum_{m=0}^{d-1} (x)_m \Tr(\alpha^m \cdot f_i^*) = (x)_i\).  Linearity is then immediate from that of \(\Tr\).
\end{proof}

\begin{proposition}[Shift property]\label{pro:shift}
For all \(x \in \F\) and all \(i \in \{0, \ldots, d-1\}\),
\[
(x)_i = (x \cdot \eta^{(i)})_0.
\]
\end{proposition}
\begin{proof}
We compute \((x \cdot \eta^{(i)})_0 = \Tr(x \cdot \eta^{(i)} \cdot f_0^*) = \Tr(x \cdot (f_i^*/f_0^*) \cdot f_0^*) = \Tr(x \cdot f_i^*) = (x)_i\).
\end{proof}

\begin{lemma}[\(\F_p\)-independence of shifting factors]\label{lem:eta-independence}
The shifting factors \(\eta^{(0)}, \eta^{(1)}, \ldots, \eta^{(d-1)}\) are \(\F_p\)-linearly independent.  In fact, they form a basis of \(\F\) over \(\F_p\).
\end{lemma}
\begin{proof}
The dual basis \(\{f_0^*, \ldots, f_{d-1}^*\}\) is a basis of \(\F\) over \(\F_p\).  The map \(f_i^* \mapsto \eta^{(i)} = f_i^*/f_0^*\) is multiplication by the nonzero scalar \(1/f_0^* \in \F^*\), hence an \(\F_p\)-linear isomorphism of \(\F\).
\end{proof}

\subsection{Shamir's secret sharing and leakage models}\label{subsec:shamir}

\begin{definition}[Shamir's scheme]\label{def:shamir}
Let \(2\le k\le n\).  Given distinct evaluation places \(x_0, \ldots, x_{n-1} \in \F^*\), the scheme \(\ShamirSS(n, k, \vec{X})_\F\) shares a secret \(s \in \F\) as follows:  sample \(P(x) = s + a_1 x + \cdots + a_{k-1} x^{k-1}\) with \(a_1, \ldots, a_{k-1} \in \F\) uniform and independent, and output shares \(s_j = P(x_j)\) for \(j = 0, \ldots, n-1\).
\end{definition}

\begin{definition}[Block leakage]\label{def:block-leakage}
In the \emph{single-block leakage} model, the adversary learns one \(\F_p\)-coordinate \((s_j)_{i_j} \in \F_p\) from each share \(s_j\), where the block indices \(\vec{i} = (i_0, \ldots, i_{n-1}) \in \{0, \ldots, d-1\}^n\) are adversary-chosen (non-adaptively).  For \(p = 2\), each \(\F_p\)-block is a single bit, so block leakage coincides with physical bit leakage.
\end{definition}

\begin{definition}[Multi-block leakage]\label{def:multi-block-leakage}
In the \emph{multi-block leakage} model, the adversary learns \(m_j \ge 1\) distinct \(\F_p\)-coordinates from share \(j\), at positions \(i_j^{(1)}, \ldots, i_j^{(m_j)} \in \{0, \ldots, d-1\}\).  The total number of blocks leaked is \(M = \sum_{j=0}^{n-1} m_j\).  A multi-block pattern is \emph{admissible} if the block positions within each share are distinct.
\end{definition}

\begin{definition}[Insecurity]\label{def:insecurity}
The insecurity of \(\ShamirSS(n, k, \vec{X})_\F\) against a leakage family \(\mc{F}\) is
\[
\eps_{\mc{F}}(\vec{X}) := \max_{f \in \mc{F}} \max_{s \in \F^*} \SD(f(\Share(0)),\; f(\Share(s))),
\]
where \(\SD\) denotes statistical distance.  We say the scheme is \emph{perfectly secure} against \(\mc{F}\) if \(\eps_{\mc{F}}(\vec{X}) = 0\).
\end{definition}

\subsection{The perfect dichotomy}\label{subsec:dichotomy}

We now record the structural foundation that the rest of the analysis builds on:  the coordinate map \((\cdot)_i = \Tr(\cdot \, \cdot f_i^*)\) is \(\F_p\)-linear, which makes the leakage map linear, forces the leakage distribution to be uniform on a coset of an \(\F_p\)-subspace, and yields a perfect dichotomy for statistical distance.

\begin{proposition}[Leakage map]\label{pro:leakage-map}
For block leakage with indices \(\vec{i}\), the leakage vector \(\vec{\ell} = ((s_j)_{i_j})_{j=0}^{n-1}\) decomposes as
\[
\ell_j = (s)_{i_j} + \sum_{t=1}^{k-1} \Tr(a_t \cdot x_j^t \cdot f_{i_j}^*),
\]
where the map \(L_{\vec{i}} : \F^{k-1} \to \F_p^n\) defined by \(L_{\vec{i}}(a_1, \ldots, a_{k-1})_j = \sum_{t=1}^{k-1} \Tr(a_t \cdot x_j^t \cdot f_{i_j}^*)\) is \(\F_p\)-linear.  (Proof in~\Cref{app:dichotomy-proofs}.)
\end{proposition}

\begin{proposition}[Perfect dichotomy]\label{pro:dichotomy}
For block leakage with indices \(\vec{i}\) and any two secrets \(s, s' \in \F\):
\[
\SD(\vec{\ell}|_s,\; \vec{\ell}|_{s'}) = \begin{cases} 0 & \tx{if } \vec{v}(s) - \vec{v}(s') \in \Img(L_{\vec{i}}) \\ 1 & \tx{if } \vec{v}(s) - \vec{v}(s') \notin \Img(L_{\vec{i}}) \end{cases}
\]
where \(\vec{v}(s) = ((s)_{i_0}, \ldots, (s)_{i_{n-1}})\).  In particular, if \(L_{\vec{i}}\) is surjective, then \(\SD = 0\) for all pairs of secrets.  (Proof in~\Cref{app:dichotomy-proofs}.)
\end{proposition}

\begin{remark}[More general result in previous work]\label{rem:MNPY-extension}
\cite[Theorem~6]{nguyen-2025} establishes \Cref{pro:dichotomy} in much greater generality:  for any linear code-based secret-sharing scheme over an extension field \(\F_{p^\lambda}\), against arbitrary physical-bit leakage.  By a remark therein, the result also applies to composite order fields \(\F_{p^d}\) and to subfield-coordinate leakage, which is our present setting.  The \(k = 2\) case is also implicit, via linear algebra, in~\cite[Theorem~5]{MNPY24}.  For completeness, proofs of \Cref{pro:leakage-map,pro:dichotomy} in our setting are deferred to~\Cref{app:dichotomy-proofs};  the rest of the paper depends only on the conclusion of~\Cref{pro:dichotomy}.
\end{remark}

\subsection{Surjectivity and the test matrix}\label{subsec:test-matrix}

\begin{proposition}[Dual characterization]\label{pro:dual}
The map \(L_{\vec{i}}\) is surjective if and only if the adjoint \(L_{\vec{i}}^* : \F_p^n \to \F^{k-1}\) is injective.  The adjoint is given by
\[
(L_{\vec{i}}^*(\vec{c}))_\ell = f_0^* \cdot \sum_{j=0}^{n-1} c_j \cdot \eta^{(i_j)} \cdot x_j^\ell, \qquad \ell = 1, \ldots, k-1.
\]
Since \(f_0^* \ne 0\), injectivity of \(L_{\vec{i}}^*\) is equivalent to:  for all \(\vec{c} \in \F_p^n \setminus \{\vec{0}\}\), there exists \(\ell \in \{1, \ldots, k-1\}\) with \(\sum_j c_j \eta^{(i_j)} x_j^\ell \ne 0\).
\end{proposition}
\begin{proof}
For an \(\F_p\)-linear map \(L : V \to W\) between finite-dimensional \(\F_p\)-vector spaces, equipped with nondegenerate \(\F_p\)-bilinear pairings, \(L\) is surjective if and only if its adjoint \(L^*\) is injective.  Indeed \(\ker(L^*) = \Img(L)^\perp\), and nondegeneracy gives \(\dim(\Img(L)^\perp) = \dim(W) - \dim(\Img(L))\), so \(\Img(L) = W\) iff \(\ker(L^*) = \{0\}\).  Both pairings below are nondegenerate. The standard one on \(\F_p^n\), and the trace pairing on \(\F^{k-1}\), because the trace form of \(\F/\F_p\) is nondegenerate.

For the formula, we compute the adjoint explicitly.  The relevant inner products are the standard \(\F_p\)-inner product \(\langle \vec{c}, \vec{y} \rangle_{\F_p} = \sum_j c_j y_j\) on \(\F_p^n\), and the trace pairing \(\langle \vec{a}, \vec{b} \rangle = \sum_{\ell=1}^{k-1} \Tr(a_\ell b_\ell)\) on \(\F^{k-1}\).  For any \(\vec{a} \in \F^{k-1}\) and \(\vec{c} \in \F_p^n\):
\begin{align*}
\langle \vec{c},\, L_{\vec{i}}(\vec{a}) \rangle_{\F_p} &= \sum_{j=0}^{n-1} c_j \sum_{t=1}^{k-1} \Tr(a_t x_j^t f_{i_j}^*) \\
&= \sum_{t=1}^{k-1} \Tr\left(a_t \sum_{j=0}^{n-1} c_j x_j^t f_{i_j}^*\right),
\end{align*}
where we used \(\F_p\)-linearity of \(\Tr\) and the fact that \(c_j \in \F_p\) can be moved inside \(\Tr\).  Now substituting \(f_{i_j}^* = \eta^{(i_j)} \cdot f_0^*\):
\begin{align*}
&= \sum_{t=1}^{k-1} \Tr\left(a_t \cdot f_0^* \cdot \sum_{j=0}^{n-1} c_j \eta^{(i_j)} x_j^t\right) = \left\langle \vec{a},\, \left(f_0^* \cdot \sum_{j=0}^{n-1} c_j \eta^{(i_j)} x_j^\ell\right)_{\ell=1}^{k-1} \right\rangle.
\end{align*}
Since this holds for all \(\vec{a}\), we conclude \((L_{\vec{i}}^*(\vec{c}))_\ell = f_0^* \cdot \sum_{j} c_j \eta^{(i_j)} x_j^\ell\).

Finally, since \(f_0^* \ne 0\) (it is a dual basis element), \(L_{\vec{i}}^*(\vec{c}) = 0\) in \(\F^{k-1}\) if and only if \(\sum_j c_j \eta^{(i_j)} x_j^\ell = 0\) for every \(\ell = 1, \ldots, k-1\).
\end{proof}

\begin{definition}[Test matrix]\label{def:test-matrix}
For block indices \(\vec{i} \in \{0, \ldots, d-1\}^n\), the \emph{test matrix} \(\Theta_{\vec{i}} \in \F_p^{d(k-1) \times n}\) is the \(\F_p\)-coordinate representation of the map
\[
\vec{c} \mapsto \left(\sum_{j=0}^{n-1} c_j \eta^{(i_j)} x_j^\ell\right)_{\ell=1}^{k-1} \in \F^{k-1} \cong \F_p^{d(k-1)}.
\]
\end{definition}

\begin{corollary}\label{cor:surj-rank}
\(L_{\vec{i}}\) is surjective if and only if \(\rank_{\F_p}(\Theta_{\vec{i}}) = n\).  Since \(\Theta_{\vec{i}}\) has only \(d(k-1)\) rows, this (sufficient) full-rank condition can hold only when \(n \le d(k-1)\), the regime in which we work.
\end{corollary}
\begin{proof}
By \Cref{pro:dual}, \(L_{\vec{i}}\) is surjective if and only if \(L_{\vec{i}}^*\) is injective, which holds if and only if the \(\F_p\)-linear map encoded by \(\Theta_{\vec{i}}\) has trivial kernel, i.e., \(\rank_{\F_p}(\Theta_{\vec{i}}) = n\).  Since \(\Theta_{\vec{i}}\) has \(d(k-1)\) rows, full column rank requires \(n \le d(k-1)\).
\end{proof}

%% ============================================================
%%  SECTION 3: THE CONSTRUCTION
%% ============================================================
\section{The construction}\label{sec:construction}

We now define the step operator and establish the properties of the orbit \(\{x_j = \Phi^j(x_0)\}\) that will be needed for the security analysis.  The required properties of this operator follow from elementary algebra, with no conditions to verify beyond \(x_0 \notin \{\alpha - 1\} \cup \Poles(n)\).

\subsection{The step operator}\label{subsec:step-operator}

\begin{definition}[Step operator]\label{def:step-operator}
Define \(\Phi : \Pp(\F) \to \Pp(\F)\) by
\[
\Phi(x) = \frac{\alpha x}{x + 1},
\]
where \(\alpha\) is the generator of \(\F^*\) from~\Cref{subsec:field}.  As a M\"obius transformation, \(\Phi\) has matrix \(M_\Phi = \begin{pmatrix} \alpha & 0 \\ 1 & 1 \end{pmatrix}\) with \(\det(M_\Phi) = \alpha \ne 0\).
\end{definition}

\begin{proposition}[Projective order]\label{pro:proj-order}
The projective order of \(\Phi\) equals \(\ord(\alpha) = p^d - 1\).
\end{proposition}
\begin{proof}
The eigenvalues of \(M_\Phi = \begin{pmatrix}\alpha & 0 \\ 1 & 1\end{pmatrix}\) are \(\alpha\) and \(1\) (reading off the diagonal, since \(M_\Phi\) is lower triangular).  Hence \(M_\Phi^j\) has eigenvalues \(\alpha^j\) and \(1\).  Now \(\Phi^j = \id\) in \(\PGL_2(\F)\) if and only if \(M_\Phi^j = \lambda I\) for some \(\lambda \in \F^*\) (since projective transformations are defined up to scalar multiples of the matrix).  The condition \(M_\Phi^j = \lambda I\) requires both eigenvalues to be equal, i.e., \(\alpha^j = 1\).  Conversely, if \(\alpha^j = 1\), then both eigenvalues equal \(1\) and \(M_\Phi^j\) must be the identity (since \(M_\Phi\) is diagonalizable over the algebraic closure:  its eigenvalues \(\alpha\) and \(1\) are distinct for \(d \ge 2\), so \(M_\Phi\), and hence \(M_\Phi^j\), is diagonalizable).  The smallest positive \(j\) with \(\alpha^j = 1\) is \(\ord(\alpha) = p^d - 1\).
\end{proof}

\begin{proposition}[Explicit formula for iterates]\label{pro:iterate-formula}
For \(j \ge 1\),
\[
\Phi^j(x) = \frac{\alpha^j x}{C_j x + 1}, \qquad C_j := \frac{\alpha^j - 1}{\alpha - 1} = \sum_{i=0}^{j-1} \alpha^i.
\]
For \(j = 0\) the formula holds with \(C_0 = 0\), giving \(\Phi^0(x) = x\).  We use this convention throughout, so that \(\Phi^0\) is the unique iterate whose pole is at \(\infty\).
\end{proposition}
\begin{proof}
By induction on \(j\).  For \(j = 1\): \(C_1 = 1\) and \(\Phi(x) = \alpha x/(x + 1)\).  For the inductive step:
\[
\Phi^{j+1}(x) = \Phi(\Phi^j(x)) = \frac{\alpha \cdot \frac{\alpha^j x}{C_j x + 1}}{\frac{\alpha^j x}{C_j x + 1} + 1} = \frac{\alpha^{j+1} x}{(\alpha^j + C_j)x + 1}.
\]
It remains to verify that \(\alpha^j + C_j = C_{j+1}\).  Indeed, \(\alpha^j + (\alpha^j - 1)/(\alpha - 1) = (\alpha^{j+1} - \alpha^j + \alpha^j - 1)/(\alpha - 1) = (\alpha^{j+1} - 1)/(\alpha - 1) = C_{j+1}\).
\end{proof}

\subsection{Distinct poles, fixed points, and pole avoidance}\label{subsec:poles-fixed-points}

\begin{corollary}[Pairwise distinct poles]\label{cor:distinct-poles}
For \(0 \le j \le n - 1\) with \(n \le p^d - 1\), the M\"obius transforms \(\Phi^0, \Phi^1, \ldots, \Phi^{n-1}\) have pairwise distinct poles in \(\Pp(\F)\).  More precisely, the pole of \(\Phi^0 = \id\) is \(\infty\), and the pole of \(\Phi^j\) for \(j \ge 1\) is \(-(\alpha - 1)/(\alpha^j - 1)\).
\end{corollary}
\begin{proof}
By \Cref{pro:iterate-formula}, the pole of \(\Phi^j\) for \(j \ge 1\) is at \(x = -1/C_j = -(\alpha - 1)/(\alpha^j - 1)\).  For \(1 \le j < j' \le n - 1\), these coincide if and only if \(\alpha^j = \alpha^{j'}\), i.e., \(j \equiv j' \pmod{p^d - 1}\).  Since \(0 < j < j' < p^d - 1\), this is impossible.  The pole of \(\Phi^0 = \id\) is \(\infty\), which is distinct from all finite poles.
\end{proof}

\begin{remark}[No involution]\label{rem:no-involution}
The map \(\Phi\) is never an involution for \(d \ge 2\), regardless of characteristic.  Indeed, \(\Phi^2 = \id\) would require \(\alpha^2 = 1\), but \(\ord(\alpha) = p^d - 1 \ge p^2 - 1 \ge 3\).  In particular, no special treatment is needed for \(p = 2\).
\end{remark}

\begin{remark}[A naive dilation fails]\label{rem:dilation-fails}
	Deriving motivation from constructions of polynomial codes with optimal list decoding properties (folded Reed-Solomon codes~\cite{guruswami-rudra-08}), the most obvious one-parameter candidate is the pure dilation~\(\Psi(x) = \alpha x\), whose iterates~\(\Psi^j(x) = \alpha^j x\) form an orbit~\(x_j = \alpha^j x_0\) on the multiplicative coset~\(x_0 \langle \alpha \rangle\).  We record here that this choice fails the structural property driving our analysis, and consequently does not admit the bad-set bound of \Cref{thm:single-block}.
	
	The functions~\(\Psi^j\) are affine, so each has its only pole at~\(\infty\).  Hence all~\(n\) iterates share a single common pole, and the pole-distinctness hypothesis of \Cref{cor:distinct-poles} fails completely (rather than failing for at most one~\(j\)).  The partial fraction argument of \Cref{sec:single-block} therefore has no effect, since the rational function~\(G_\ell(x) = \sum_j c_j \eta^{(i_j)} (\alpha^j x)^\ell = x^\ell \cdot \sum_j c_j \eta^{(i_j)} \alpha^{j \ell}\) is a monomial in~\(x\), and~\(G_\ell \equiv 0\) is equivalent to the \(x_0\)-independent condition~\(\sum_j c_j \eta^{(i_j)} \alpha^{j \ell} = 0\).  The starting point~\(x_0\) has dropped out of the rank condition entirely.  In other words, the test matrix~\(\Theta_{\vec{i}}\) is the same matrix for every~\(x_0 \in \F^*\), so the dilation scheme is either secure for all nonzero~\(x_0\) or for none -- there is no parameter to optimize and no closed-form bad-set bound to prove.  In particular, the central quantitative result~\(|\Bad| \le n + n(dp)^n\) of \Cref{thm:single-block} has no analoguefor~\(\Psi\).
	
	The role of the~\(+1\) in the denominator of~\(\Phi\) is thus to perturb each nontrivial iterate's pole away from~\(\infty\) to a distinct finite point of~\(\F^*\) (\Cref{cor:distinct-poles}), placing the orbit in generic position relative to the partial fraction decomposition.  This is the essential algebraic feature that distinguishes~\(\Phi\) from its affine cousin and makes the closed-form bad-set bound possible.
\end{remark}

\begin{lemma}[Fixed points]\label{lem:fixed-points}
For each \(1 \le j < p^d - 1\), the fixed points of \(\Phi^j\) in \(\Pp(\F)\) are exactly \(\{0, \alpha - 1\}\).  In particular, \(\ShortPeriod(n) := \{x_0 \in \F^* : \exists\, 1 \le j < n,\; \Phi^j(x_0) = x_0\} = \{\alpha - 1\}\).
\end{lemma}
\begin{proof}
Fix \(1 \le j < p^d - 1\), so that \(\Phi^j \ne \id\) by \Cref{pro:proj-order}.  The fixed-point equation \(\Phi^j(x) = x\) reads \(\alpha^j x / (C_j x + 1) = x\).  For \(x = 0\): \(\Phi^j(0) = 0\), so \(0\) is a fixed point, but \(0 \notin \F^*\).  For \(x \ne 0\):  dividing both sides by \(x\) gives \(\alpha^j/(C_j x + 1) = 1\), hence \(x = (\alpha^j - 1)/C_j\).  Since \(C_j = (\alpha^j - 1)/(\alpha - 1)\), this gives \(x = \alpha - 1\).  For \(x = \infty\): \(\Phi^j(\infty) = \alpha^j / C_j \ne \infty\) (since \(C_j \ne 0\) for \(j < p^d - 1\)).  The fixed-point set \(\{0, \alpha - 1\}\) is independent of \(j\), so restricting to \(\F^*\) gives \(\ShortPeriod(n) = \{\alpha - 1\}\).
\end{proof}

\begin{definition}[Pole set]\label{def:pole-set}
Define \(\Poles(n) := \{-1/C_j : 1 \le j \le n-1\} \subset \F^*\).  These are pairwise distinct by \Cref{cor:distinct-poles}, so \(|\Poles(n)| = n - 1\).

We note that \(\Poles(n) \cap \{\alpha - 1\} = \emptyset\):  the equality \(\alpha - 1 = -(\alpha-1)/(\alpha^j - 1)\) would require \(\alpha^j = 0\), which is impossible.
\end{definition}

\begin{proposition}[Distinct evaluation places]\label{pro:distinct-eval}
For \(x_0 \in \F^* \setminus (\{\alpha - 1\} \cup \Poles(n))\) and \(n \le p^d - 1\), the evaluation places \(x_0, x_1, \ldots, x_{n-1}\) are pairwise distinct elements of \(\F^*\).
\end{proposition}
\begin{proof}
We verify three properties.  \n(a) \emph{Each \(x_j \in \F^*\):}  \(x_j = 0\) if and only if \(x_0 = 0\) (since \(0\) is a fixed point of \(\Phi\)), which is excluded by \(x_0 \in \F^*\).  \(x_j = \infty\) if and only if \(x_0 = -1/C_j \in \Poles(n)\), which is excluded by hypothesis.  \n(b) \emph{The orbit avoids \(\alpha - 1\):}  Since \(\alpha - 1\) is a fixed point of \(\Phi\) itself (not just of \(\Phi^j\)), we have \(\Phi^{-1}(\{\alpha - 1\}) = \{\alpha - 1\}\), and hence \(x_j = \alpha - 1\) if and only if \(x_0 = \alpha - 1\), which is excluded.  \n(c) \emph{Pairwise distinct:}  \(x_j = x_{j'}\) for \(j < j'\) if and only if \(x_j\) is a fixed point of \(\Phi^{j'-j}\) in \(\F^*\).  By \Cref{lem:fixed-points}, the only such point is \(\alpha - 1\), and by~(b), \(x_j \ne \alpha - 1\).
\end{proof}

\begin{remark}\label{rem:structural-exclusions}
The total number of structural exclusions is \(|\{\alpha - 1\} \cup \Poles(n)| = n\).
\end{remark}

\begin{remark}[Canonical base point]\label{rem:alpha-starting-point}
The generator \(\alpha\) is a natural canonical candidate for the base point.  Two layers of the goodness condition should be distinguished.

\emph{Provably:}  \(\alpha\) avoids the structural exclusion set \(\{\alpha - 1\} \cup \Poles(n)\).  Indeed, \(\alpha \ne \alpha - 1\) (since \(1 \ne 0\)), and \(\alpha \in \Poles(n)\) would require \(\alpha^{j+1} = 1\) for some \(1 \le j \le n - 1\), which is impossible whenever \(n \le p^d - 2\).  Verification is \(O(n)\).

\emph{Not provably without verification:}  whether \(\alpha\) avoids the security-bad set in the headline single-block regime of \Cref{thm:single-block} requires running the classifier of \Cref{alg:classifier}.
\end{remark}

\subsection{Classifier and choice of base point}\label{subsec:classifier}

We now give an explicit classifier that, given a candidate base point \(x_0\), tests the \emph{sufficient} full-rank condition of~\Cref{cor:surj-rank} for the resulting evaluation places \(x_j = \Phi^j(x_0)\):  for each block-index tuple \(\vec{i} \in \{0, \ldots, d-1\}^n\) it checks whether the test matrix \(\Theta_{\vec{i}}\) has full column rank \(n\) over \(\F_p\), returning \textsc{good} only if all of them do.  A \textsc{good} verdict \emph{certifies} perfect security against single-block leakage, so the classifier is a \emph{sound} tester.  A \textsc{bad} \((\vec{i})\) verdict reports a pattern at which the full-rank certificate fails; since full rank is sufficient but not necessary (\Cref{pro:dichotomy}), this does not by itself certify insecurity.  On the structured M\"obius candidates we propose, this sound test is what we use to certify goodness; we do not claim it decides perfect security on arbitrary evaluation places, for which the exact criterion is the containment/minimal-codeword condition of~\cite[Theorem~6]{nguyen-2025}.

\begin{algorithm}[h]
\DontPrintSemicolon
\KwIn{prime \(p\), integer \(d \ge 2\), reconstruction threshold \(k \ge 2\), number of parties \(n \le d(k-1)\), field representation \(\F = \F_p[\alpha]\), candidate base point \(x_0 \in \F^*\).}
\KwOut{\textsc{good} (certifies perfect security), or \textsc{bad}\((\vec{i})\) where \(\vec{i}\) is a pattern at which the sufficient full-rank certificate fails.}
\BlankLine
\If{\(x_0 \in \{\alpha - 1\} \cup \Poles(n)\)}{\Return \textsc{bad}\((\bot)\) \tcp*{structural exclusion}}
Compute \(x_j \gets \Phi^j(x_0)\) for \(j = 0, \ldots, n - 1\)\;
Compute the dual basis \(\{f_0^*, \ldots, f_{d-1}^*\}\) and shifting factors \(\eta^{(i)} = f_i^* / f_0^*\)\;
\For{\(\vec{i} \in \{0, \ldots, d-1\}^n\)}{
    Build \(\Theta_{\vec{i}} \in \F_p^{d(k-1) \times n}\) as in \Cref{def:test-matrix}\;
    \If{\(\rank_{\F_p}(\Theta_{\vec{i}}) < n\)}{\Return \textsc{bad}\((\vec{i})\)}
}
\Return \textsc{good}\;
\caption{\(\textsc{Classify}(p, d, k, n, x_0)\): single-block-leakage classifier for the M\"obius orbit \(x_j = \Phi^j(x_0)\).}\label{alg:classifier}
\end{algorithm}

\begin{proposition}[Classifier complexity]\label{pro:classifier-complexity}
\Cref{alg:classifier} performs \(O(d^n \cdot n^2 \cdot dk)\) arithmetic operations over \(\F_p\), after preprocessing costing \(\poly(d, n, k, \log p)\) operations.
\end{proposition}
\begin{proof}
The structural exclusion check takes \(O(n)\) time.

\emph{Preprocessing.}\quad  Compute the orbit \(x_0, \ldots, x_{n-1}\), the dual basis, and the shifting factors, and then tabulate \emph{once} the \(dn(k-1)\) field elements \(\eta^{(i)} x_j^\ell\) (\(0 \le i \le d-1\), \(0 \le j \le n-1\), \(1 \le \ell \le k-1\)) together with their \(\F_p\)-coordinate vectors.  This costs \(O(dnk)\) multiplications in \(\F\), i.e. \(\poly(d, n, k, \log p)\) operations over \(\F_p\), and is performed outside the main loop.

\emph{Main loop.}\quad  Iterate \(d^n\) times.  Each iteration assembles \(\Theta_{\vec{i}}\) by table lookup in \(O(ndk)\) operations over \(\F_p\), and computes its rank by Gaussian elimination in \(O(n^2 \cdot d(k-1)) = O(n^2 dk)\) operations over \(\F_p\).  No enumeration over coefficient vectors \(\vec{c} \in \F_p^n\) is needed, since the rank test handles all \(\vec{c}\) per pattern simultaneously.
\end{proof}

\Cref{alg:classifier} matches the cost of the classifying algorithm given by Nguyen~\cite[Fig.~1]{nguyen-2025} on the corresponding inputs.  Our contribution in this regard is therefore not a faster verification algorithm; rather, the classifier exists here as a tool to certify the structured Möbius candidates we propose, complementing the closed form bad-set bounds.

%\begin{remark}[Choosing \(x_0\) in practice]\label{rem:choosing-x0}
%Three options for the base point are available, with different verification costs.
%
%\emph{(a) Sub-regime \(n \le k - 1\):}  take \(x_0 = \alpha\) with no further check, by \Cref{cor:alpha-good-low-n}.  Verification is \(O(n)\) (only the structural-exclusion check is needed).
%
%\emph{(b) Headline regime, canonical candidate:}  take \(x_0 = \alpha\) and verify with~\Cref{alg:classifier}.  Whether \(\alpha\) is in the bad set is deterministic, depending on the choice of primitive polynomial \(\Pi\); but the bad-set bound \(|\Bad| \le n + n(dp)^n \ll p^d - 1\) (\Cref{thm:single-block}) shows that the bad set is a thin subset of \(\F^*\), so \(\alpha\) is good for most choices of \(\Pi\).  If \Cref{alg:classifier} returns \textsc{bad}, a different primitive polynomial can be chosen.
%
%\emph{(c) Headline regime, random candidate:}  pick \(x_0 \in \F^*\) uniformly at random and verify with~\Cref{alg:classifier}.  Verification succeeds with probability at least \(1 - (n + n(dp)^n)/(p^d - 1)\), which is \(1 - o(1)\) in our parameter range.
%
%Rejection sampling is \emph{not} the proposed strategy:  options~(a) and~(b) are deterministic recipes; option~(c) is offered as an alternative.  In all three cases, once \(x_0\) is fixed, the evaluation places \(x_j = \Phi^j(x_0)\) are determined by the field representation alone.
%\end{remark}

%% ============================================================
%%  SECTION 4: SINGLE-BLOCK SECURITY
%% ============================================================
\section{Single-block security}\label{sec:single-block}

The goal of this section is to prove that for all but a small number of base points \(x_0\), the test matrix \(\Theta_{\vec{i}}\) has full column rank for \emph{every} block-index tuple \(\vec{i}\).  The main tool is a partial fraction nondegeneracy lemma, which exploits the fact that the M\"obius transforms \(\Phi^0, \ldots, \Phi^{n-1}\) have pairwise distinct poles.  Throughout the section we write, for \(\vec{i} \in \{0, \ldots, d-1\}^n\), \(\vec{c} \in \F_p^n\) and \(\ell \ge 1\),
\[
G_\ell(x; \vec{i}, \vec{c}) := \sum_{j=0}^{n-1} c_j \, \eta^{(i_j)} \, (\Phi^j(x))^\ell \;\in\; \F(x),
\]
which is the function \(G_\ell\) of \Cref{lem:nondegeneracy} specialized to \(\psi_j = \Phi^j\) and \(\delta_j = \eta^{(i_j)}\).

\subsection{Nondegeneracy via partial fractions}\label{subsec:nondegeneracy}

\begin{lemma}[nondegeneracy]\label{lem:nondegeneracy}
Let \(\psi_0, \ldots, \psi_{n-1}\) be M\"obius transforms on \(\Pp(\F)\) with pairwise distinct poles.  Let \(\delta_0, \ldots, \delta_{n-1} \in \F^*\), \(c_0, \ldots, c_{n-1} \in \F_p\) with \(\vec{c} \ne \vec{0}\), and \(\ell \ge 1\).  Then
\[
G_\ell(x) := \sum_{j=0}^{n-1} c_j \delta_j (\psi_j(x))^\ell
\]
is not identically zero as a rational function on \(\Pp(\F)\).
\end{lemma}
\begin{proof}
Write \(\psi_j(x) = (A_j x + B_j)/(C_j x + D_j)\) with \(A_j D_j - B_j C_j \ne 0\).  The poles \(p_j = -D_j/C_j\) (for \(C_j \ne 0\)) or \(p_j = \infty\) (for \(C_j = 0\)) are pairwise distinct by hypothesis, so at most one equals \(\infty\).

\vst
\noindent\emph{Case~A}\quad All poles are finite (\(C_j \ne 0\) for all \(j\)).  Near the pole \(p_j\), \(\psi_j(x)\) has a simple pole with residue \(r_j = -(A_j D_j - B_j C_j)/C_j^2 \ne 0\).  Consequently, \((\psi_j(x))^\ell\) has a pole of order \(\ell\) at \(p_j\) with leading Laurent coefficient \(r_j^\ell \ne 0\).  Since the poles are pairwise distinct, the leading term of \(G_\ell\) at \(p_j\) comes only from the \(j\)-th summand: \(G_\ell(x) \sim c_j \delta_j r_j^\ell / (x - p_j)^\ell\) as \(x \to p_j\).  If \(G_\ell \equiv 0\), then \(c_j \delta_j r_j^\ell = 0\) for all \(j\).  Since \(\delta_j \ne 0\) and \(r_j \ne 0\), this forces \(c_j = 0\) for all \(j\), contradicting \(\vec{c} \ne \vec{0}\).

\vst
\noindent\emph{Case~B}\quad One pole is at \(\infty\) (say \(C_0 = 0\), so that \(\psi_0\) is affine; all other \(C_j \ne 0\)).  The residue argument at each finite pole \(p_j\) (\(j \ge 1\)) gives \(c_j \delta_j r_j^\ell = 0\), hence \(c_j = 0\) for all \(j \ge 1\).  Then \(G_\ell(x) = c_0 \delta_0 (\psi_0(x))^\ell\).  Since \(\psi_0\) is nonconstant (\(A_0/D_0 \ne 0\) because \(\det(M_0) \ne 0\) and \(C_0 = 0\)), \(\delta_0 \ne 0\), and \(c_0 \ne 0\) (as \(\vec{c} \ne \vec{0}\)), this is not identically zero.

At most one pole equals \(\infty\), so Cases~A and~B are exhaustive.
\end{proof}

\begin{corollary}[Zero bound]\label{cor:zero-bound}
Under the hypotheses of \Cref{lem:nondegeneracy}, \(G_\ell\) has at most \(n\ell\) zeros in \(\Pp(\F)\).
\end{corollary}
\begin{proof}
Write \(\psi_j(x) = (A_j x + B_j)/(C_j x + D_j)\).  Let \(S = \{j : C_j \ne 0\}\) and \(T = \{j : C_j = 0\}\), so \(|T| \le 1\) (at most one pole can be \(\infty\)).  Define the common denominator
\[
D(x) = \prod_{j \in S} (C_j x + D_j)^\ell,
\]
and set \(N(x) = G_\ell(x) \cdot D(x)\).  We compute the degree of \(N(x)\).

For \(j \in S\):  the \(j\)-th summand of \(G_\ell\) is \(c_j \delta_j (A_j x + B_j)^\ell / (C_j x + D_j)^\ell\).  After multiplying by \(D(x)\), the denominator \((C_j x + D_j)^\ell\) cancels, contributing to \(N(x)\) the term \(c_j \delta_j (A_j x + B_j)^\ell \prod_{j' \in S,\, j' \ne j} (C_{j'} x + D_{j'})^\ell\), of degree at most \(\ell + \ell(|S| - 1) = \ell |S|\).

For \(j \in T\):  we have \(C_j = 0\), so \(\psi_j(x) = (A_j/D_j) x + B_j/D_j\) is affine (with \(A_j/D_j \ne 0\)).  The \(j\)-th summand of \(G_\ell\) is \(c_j \delta_j ((A_j/D_j) x + B_j/D_j)^\ell\), a polynomial of degree \(\ell\).  After multiplying by \(D(x) = \prod_{j' \in S} (C_{j'} x + D_{j'})^\ell\), which has degree \(\ell |S|\), the contribution is a polynomial of degree \(\ell + \ell |S|\).

Since \(|S| + |T| = n\):  if \(T = \emptyset\), then \(|S| = n\) and \(\deg(N) \le \ell n\).  If \(|T| = 1\), then \(|S| = n - 1\) and \(\deg(N) \le \max(\ell(n-1),\, \ell + \ell(n-1)) = \ell n\).  In both cases \(\deg(N) \le n\ell\), and \(\deg(D) = \ell|S| \le n\ell\).

Since \(G_\ell \not\equiv 0\) by \Cref{lem:nondegeneracy}, we have \(N \not\equiv 0\).  Every zero of \(G_\ell\) lying in \(\F\) is a root of \(N\) (evaluate \(N = G_\ell D\) there), so \(G_\ell\) has at most \(\deg(N) \le n\ell\) zeros in \(\F\) --- which is all we use below.  For the count in \(\Pp(\F)\), recall that a nonzero rational function \(R = N/D \in \F(x)\) has exactly \(\deg R = \max(\deg N, \deg D)\) zeros in \(\Pp(\ol{\F})\), counted with multiplicity;  here \(\deg N \le n\ell\) and \(\deg D = \ell|S| \le n\ell\), so \(G_\ell\) has at most \(n\ell\) zeros in \(\Pp(\F)\) as well.
\end{proof}

\subsection{Full rank of the test matrix}\label{subsec:full-rank}

\begin{theorem}\label{thm:full-rank-single}
Let \(p\) be any prime, \(d \ge 2\), and \(2 \le k \le n \le d(k-1)\) with \(n \le p^d - 1\).  For any \(\vec{i} \in \{0, \ldots, d-1\}^n\) and \(\vec{c} \in \F_p^n \setminus \{\vec{0}\}\), we have
\[
|\{x_0 \in \F^* \setminus \Poles(n) : \sum_{j=0}^{n-1} c_j \eta^{(i_j)} (\Phi^j(x_0))^\ell = 0 \;\;\forall\; \ell = 1, \ldots, k-1\}| \le n.
\]
\end{theorem}
\begin{proof}
Apply \Cref{lem:nondegeneracy} with \(\ell = 1\), \(\psi_j = \Phi^j\), and \(\delta_j = \eta^{(i_j)}\).  By \Cref{cor:distinct-poles}, the \(\Phi^j\) have pairwise distinct poles.  The shifting factors \(\eta^{(i_j)}\) are nonzero (being ratios of dual basis elements).  So \(G_1(x) = \sum_j c_j \eta^{(i_j)} \Phi^j(x) \not\equiv 0\), and by \Cref{cor:zero-bound}, \(G_1\) has at most \(n\) zeros in \(\Pp(\F)\).  The bad set is contained in \(\{x_0 \in \F^* : G_1(x_0) = 0\}\), which has size \(\le n\).
\end{proof}

\begin{remark}\label{rem:ell-equals-one}
The bound \(\le n\) comes from the \(\ell = 1\) case alone.  The bad set requires \(G_\ell = 0\) for \emph{all} \(\ell\), and the containment in the zero set of \(G_1\) already suffices.
\end{remark}

\begin{theorem}[Perfect security against single-block leakage]\label{thm:single-block}
Let \(p\) be any prime, \(d \ge 2\), and \(2 \le k \le n \le d(k-1)\) with \(n \le p^d - 1\).  Define
\[
\Bad := \{\alpha - 1\} \cup \Poles(n) \;\cup\; \bigcup_{\vec{i} \in \{0,\ldots,d-1\}^n} \; \bigcup_{\vec{c} \in \F_p^n \setminus \{\vec{0}\}} \{x_0 \in \F^* : G_1(x_0; \vec{i}, \vec{c}) = 0\}.
\]
Then \(|\Bad| \le n + n(dp)^n\).  For any \(x_0 \in \F^* \setminus \Bad\):
\begin{enumerate}[(a)]
	\item  \(x_0, \ldots, x_{n-1}\) are pairwise distinct elements of \(\F^*\).
	\item  \(\ShamirSS(n, k, \vec{X})_\F\) is perfectly secure against all single-block leakage.
\end{enumerate}
\end{theorem}
\begin{proof}
Part~(a) follows from \Cref{pro:distinct-eval}, since \(x_0 \notin \{\alpha - 1\} \cup \Poles(n)\).

For part~(b):  by the security reduction (\Cref{cor:surj-rank}), we need \(\rank_{\F_p}(\Theta_{\vec{i}}) = n\) for every block-index tuple \(\vec{i} \in \{0, \ldots, d-1\}^n\).  By \Cref{pro:dual}, this fails for \(\vec{i}\) if and only if there exists \(\vec{c} \in \F_p^n \setminus \{\vec{0}\}\) with \(\sum_j c_j \eta^{(i_j)} x_j^\ell = 0\) for all \(\ell = 1, \ldots, k-1\).  For each such pair \((\vec{i}, \vec{c})\), \Cref{thm:full-rank-single} bounds the set of bad base points by \(\le n\).

We now count.  The structural exclusion set \(\{\alpha - 1\} \cup \Poles(n)\) has size \(n\) (by \Cref{rem:structural-exclusions}).  The number of block-index tuples is \(|\{0, \ldots, d-1\}^n| = d^n\).  The number of nonzero coefficient vectors is \(|\F_p^n \setminus \{\vec{0}\}| = p^n - 1\).  By \Cref{thm:full-rank-single}, each pair contributes at most \(n\) bad base points.  By the union bound:
\[
|\Bad| \le n + d^n \cdot (p^n - 1) \cdot n \le n + n \cdot d^n \cdot p^n = n + n(dp)^n.
\]

For any \(x_0 \in \F^* \setminus \Bad\):  by construction, \(x_0\) avoids the structural exclusion set (ensuring part~(a)) and, for every \(\vec{i}\) and every nonzero \(\vec{c}\), there exists \(\ell\) with \(G_\ell(x_0; \vec{i}, \vec{c}) \ne 0\).  This means \(\Theta_{\vec{i}}\) has trivial kernel for every \(\vec{i}\), i.e., full column rank \(n\).  By \Cref{cor:surj-rank}, \(L_{\vec{i}}\) is surjective for all \(\vec{i}\), and by \Cref{pro:dichotomy}, \(\SD = 0\) for all pairs of secrets and all block-leakage patterns.
\end{proof}

\begin{corollary}[Existence]\label{cor:existence-single}
A good \(x_0\) exists whenever \(p^d - 1 > n + n(dp)^n\); a sufficient condition is \(d > n\big(1 + \log_p d\big) + \log_p(2n)\), which gives the range \(n = O\big(d/\log_p(pd)\big) = O\big(d/(1 + \log_p d)\big)\).
\end{corollary}

\subsection{Physical bit security}\label{subsec:physical-bit}

\begin{proposition}[Block security implies sub-block security]\label{pro:sub-block}
Let \(S_0, \ldots, S_{n-1}\) be arbitrary nonempty finite sets, assume the hypotheses of \Cref{thm:single-block}, and let \(x_0 \in \F^* \setminus \Bad\).  Then \(\ShamirSS(n, k, \vec{X})_\F\) is perfectly secure against any leakage that applies an arbitrary function \(g_j : \F_p \to S_j\) to a single \(\F_p\)-block of each share.
\end{proposition}
\begin{proof}
By \Cref{thm:single-block}, for every block-index tuple \((i_0, \ldots, i_{n-1}) \in \{0, \ldots, d-1\}^n\), the map \(L_{\vec{i}}\) is surjective.  By \Cref{pro:dichotomy}, this means the leakage vector \(((P(X_0))_{i_0}, \ldots, (P(X_{n-1}))_{i_{n-1}}) \in \F_p^n\) has the same distribution for every secret \(s \in \F\):  it is uniform over \(\F_p^n\).

Now suppose the adversary applies deterministic functions \(g_j : \F_p \to S_j\) to obtain\\\((g_0((P(X_0))_{i_0}), \ldots, g_{n-1}((P(X_{n-1}))_{i_{n-1}}))\).  Since each \(g_j\) is a fixed function, this post-processed leakage is a deterministic function of \(((P(x_j))_{i_j})_{j=0}^{n-1}\).  By the data processing inequality for statistical distance, for any two secrets \(s, s'\):
\[
\SD\big((g_j((P(x_j))_{i_j}))_j \,\big|\, s,\;\; (g_j((P(x_j))_{i_j}))_j \,\big|\, s'\big) \le \SD\big(((P(x_j))_{i_j})_j \,\big|\, s,\;\; ((P(x_j))_{i_j})_j \,\big|\, s'\big) = 0.
\]
Hence perfect security is preserved under arbitrary post-processing.
\end{proof}

\begin{corollary}[Physical bit security, all characteristics]\label{cor:physical-bit}
Assume the hypotheses of \Cref{thm:single-block} and let \(x_0 \in \F^* \setminus \Bad\).  For any prime \(p\), if the adversary leaks one physical bit per share (from an arbitrary \(\F_p\)-block, possibly a different block for each share), then \(\ShamirSS(n, k, \vec{X})_\F\) is perfectly secure.
\end{corollary}
\begin{proof}
Each share is stored as \(d\) blocks over \(\F_p\), each block as a \(\lceil \log_2 p \rceil\)-bit string, so a physical bit \(\mathrm{bit}_{r_j}((P(x_j))_{i_j})\) is a function \(g_j : \F_p \to \{0,1\}\) of the \emph{single} block \((P(x_j))_{i_j}\).  By \Cref{thm:single-block}, the scheme is perfectly secure against single-block leakage, so the claim follows from \Cref{pro:sub-block}.
\end{proof}
%
%\begin{remark}\label{rem:cross-block}
%This resolves the derandomization of~\cite{MNPY24} for the single-bit-per-share regime, across all characteristics.  The cross-block physical bit case for \(p > 2\), where the bit-extraction map is nonlinear, is handled by the data processing inequality, since the block-level security is perfect.
%\end{remark}

%% ============================================================
%%  SECTION 5: MULTI-BLOCK SECURITY
%% ============================================================
\section{Multi-block security}\label{sec:multi-block}

We now extend the analysis to the multi-block setting, where the adversary probes multiple \(\F_p\)-coordinates from each share.

\subsection{Setup}\label{subsec:multi-block-setup}

For an admissible multi-block leakage pattern with \(M = \sum_{j=0}^{n-1} m_j \le d(k-1)\) total blocks, index the leaked coordinates by the \(M\) pairs \((j, r)\), where \(0 \le j \le n-1\) and \(1 \le r \le m_j\).  The development of \Cref{subsec:dichotomy,subsec:test-matrix} goes through with the single index \(j\) replaced by the pair \((j,r)\), the block index \(i_j\) replaced by \(i_j^{(r)}\), and \(n\) replaced by \(M\).  We record the outcome.

\begin{proposition}[Multi-block reduction]\label{pro:multi-block-reduction}
	With the pattern as above, write \(\vec{\ell} = \big((s_j)_{i_j^{(r)}}\big)_{j,r} \in \F_p^M\).  Then:
	\begin{enumerate}[(a)]
		\item  \(\ell_{j,r} = (s)_{i_j^{(r)}} + L(\vec{a})_{j,r}\), where \(L : \F^{k-1} \to \F_p^M\), given by \(L(\vec{a})_{j,r} = \sum_{t=1}^{k-1} \Tr(a_t \cdot x_j^t \cdot f_{i_j^{(r)}}^*)\), is \(\F_p\)-linear;
		\item  \(\SD(\vec{\ell}|_s,\, \vec{\ell}|_{s'}) \in \{0, 1\}\) for all \(s, s' \in \F\), and \(\SD = 0\) whenever \(L\) is surjective;
		\item  \(L\) is surjective if and only if the test matrix \(\Theta \in \F_p^{d(k-1) \times M}\), defined as the \(\F_p\)-coordinate representation of
		\[
		\vec{c} \mapsto \bigg(\sum_{j=0}^{n-1} \sum_{r=1}^{m_j} c_{j,r} \, \eta^{(i_j^{(r)})} x_j^\ell\bigg)_{\ell=1}^{k-1} \in \F^{k-1} \cong \F_p^{d(k-1)},
		\]
		has full column rank \(M\).
	\end{enumerate}
\end{proposition}
\begin{proof}
	The proofs of \Cref{pro:leakage-map,pro:dichotomy,pro:dual} and \Cref{cor:surj-rank} use only  (i) the \(\F_p\)-linearity of the coordinate maps \((\cdot)_i\) (\Cref{fac:coord-trace}), (ii) the \(\F\)-linearity of \(\vec{a} \mapsto (P(x_j))_j\) in the randomness, and (iii) \(f_i^* = \eta^{(i)} f_0^*\) with \(f_0^* \ne 0\) (\Cref{def:dual-basis}).  None of these depends on the leaked coordinates being one per share, so every step applies verbatim after replacing the index set \(\{0, \ldots, n-1\}\) by \(\{(j,r)\}\) and \(n\) by \(M\).
\end{proof}

Full column rank \(\rank_{\F_p}(\Theta) = M\) is thus a \emph{sufficient} condition for perfect security, and is the criterion we establish below; as in the single-block case (\Cref{pro:dichotomy}) it need not be necessary.

The structure of the kernel is as follows.  A nonzero vector \((c_{j,r})_{j,r} \in \ker(\Theta)\) groups naturally by shares, producing \(d_j = \sum_{r=1}^{m_j} c_{j,r} \eta^{(i_j^{(r)})} \in \F\).  By admissibility, the shifting factors within each share are \(\F_p\)-independent, so \(d_j = 0\) if and only if all of the corresponding \(c_{j,r}\) vanish.  In particular, at least one \(d_j\) must be nonzero.

\begin{remark}[Admissibility gives within-share independence]\label{rem:admissibility}
	Probing the same block position of a share twice returns the same symbol, so every pattern is equivalent to an admissible one, with \(m_j \le d\).  What admissibility buys is exactly the fact used in \Cref{thm:multi-block-fixed-pattern}, which is that the shifting factors \(\eta^{(i_j^{(1)})}, \ldots, \eta^{(i_j^{(m_j)})}\) attached to a single share are distinct members of the \(\F_p\)-basis \(\{\eta^{(0)}, \ldots, \eta^{(d-1)}\}\) (\Cref{lem:eta-independence}), hence \(\F_p\)-linearly independent, so that \(d_j = \sum_{r} c_{j,r} \eta^{(i_j^{(r)})} = 0\) forces \(c_{j,1} = \cdots = c_{j,m_j} = 0\).  Without distinctness, \(\Theta\) has repeated columns and the full-rank certificate fails for trivial reasons.
\end{remark}

\subsection{Nondegeneracy with field coefficients}\label{subsec:field-coeff}

\begin{lemma}\label{lem:nondegeneracy-F}
Let \(\psi_0, \ldots, \psi_{n-1}\) be M\"obius transforms on \(\Pp(\F)\) with pairwise distinct poles.  Let \(d_0, \ldots, d_{n-1} \in \F\) with at least one \(d_j \ne 0\), and \(\ell \ge 1\).  Then \(\sum_j d_j (\psi_j(x))^\ell \not\equiv 0\) as a rational function, and has at most \(n\ell\) zeros in \(\Pp(\F)\).
\end{lemma}
\begin{proof}
Write \(\psi_j(x) = (A_j x + B_j)/(C_j x + D_j)\) with \(\det(M_j) = A_j D_j - B_j C_j \ne 0\), and let \(p_j\) denote the pole of \(\psi_j\).  The poles are pairwise distinct by hypothesis.  Let \(S = \{j : C_j \ne 0\}\) (finite poles) and \(T = \{j : C_j = 0\}\) (pole at \(\infty\)), so \(|T| \le 1\).

\vst
\noindent\emph{Case~A}\quad All poles are finite (\(T = \emptyset\)).  For each \(j\), the residue of \(\psi_j\) at its pole \(p_j = -D_j/C_j\) is \(r_j = -\det(M_j)/C_j^2 \ne 0\).  Near \(p_j\), we have \(\psi_j(x) = r_j/(x - p_j) + O(1)\), so \((\psi_j(x))^\ell = r_j^\ell/(x - p_j)^\ell + \tx{lower-order terms}\).  Since \(p_j \ne p_{j'}\) for \(j \ne j'\), the function \(\sum_j d_j (\psi_j(x))^\ell\) has a Laurent expansion near \(p_j\) whose leading term is \(d_j r_j^\ell / (x - p_j)^\ell\); all other summands are regular at \(p_j\).  If \(\sum_j d_j (\psi_j(x))^\ell \equiv 0\), this leading coefficient must vanish: \(d_j r_j^\ell = 0\).  Since \(r_j \ne 0\), we get \(d_j = 0\) for every \(j \in S\).  With \(T = \emptyset\), this gives \(d_j = 0\) for all \(j\), contradicting the hypothesis.

\vst
\noindent\emph{Case~B}\quad One pole is at \(\infty\) (say \(j = 0\), so \(C_0 = 0\); all other \(C_j \ne 0\)).  The residue argument at each finite pole \(p_j\) (\(j \ge 1\)) gives \(d_j r_j^\ell = 0\), hence \(d_j = 0\) for all \(j \ge 1\).  Then \(\sum_j d_j (\psi_j(x))^\ell = d_0 (\psi_0(x))^\ell\).  Since \(C_0 = 0\), \(\psi_0(x) = (A_0/D_0) x + B_0/D_0\) is a nonconstant affine function (\(A_0/D_0 \ne 0\) because \(\det(M_0) = A_0 D_0 \ne 0\)).  Thus \(d_0 (\psi_0(x))^\ell\) is a nonzero polynomial (since \(d_0 \ne 0\) by hypothesis, as all other \(d_j\) vanish).

The zero bound \(\le n\ell\) follows by the same degree argument as in \Cref{cor:zero-bound}:  multiply \(\sum_j d_j (\psi_j(x))^\ell\) by \(D(x) = \prod_{j \in S} (C_j x + D_j)^\ell\) to obtain a nonzero polynomial \(N(x)\) of degree \(\le n\ell\).
\end{proof}

\subsection{Multi-block theorems}\label{subsec:multi-block-theorems}

\begin{theorem}[Multi-block, fixed pattern]\label{thm:multi-block-fixed-pattern}
Let \(p\) be any prime, \(d \ge 2\), and \(2 \le k \le n \le d(k-1)\) with \(n \le p^d - 1\).  Fix an admissible multi-block leakage pattern with \(M \le d(k-1)\) total blocks.  The set of \(x_0 \in \F^*\) for which the scheme is not perfectly secure \emph{against this pattern} has size \(\le n + n \cdot p^M\).  A good \(x_0\) exists when \(M < d - \log_p(2n)\).
\end{theorem}
\begin{proof}
A nonzero vector \((c_{j,r})_{j,r} \in \F_p^M \setminus \{\vec{0}\}\) lies in \(\ker(\Theta)\) if and only if
\[
\sum_{j=0}^{n-1} \sum_{r=1}^{m_j} c_{j,r} \eta^{(i_j^{(r)})} x_j^\ell = 0 \quad \tx{for all } \ell = 1, \ldots, k-1.
\]
Grouping the inner sum by shares, this becomes \(\sum_{j=0}^{n-1} d_j x_j^\ell = 0\) for all \(\ell\), where \(d_j := \sum_{r=1}^{m_j} c_{j,r} \eta^{(i_j^{(r)})} \in \F\).  By admissibility, the shifting factors \(\eta^{(i_j^{(1)})}, \ldots, \eta^{(i_j^{(m_j)})}\) within share \(j\) are \(\F_p\)-independent (\Cref{lem:eta-independence}), so \(d_j = 0\) if and only if \(c_{j,r} = 0\) for all \(r\).  Since \(\vec{c} \ne \vec{0}\), at least one \(d_j \ne 0\).

By \Cref{lem:nondegeneracy-F} with \(\ell = 1\) and \(\psi_j = \Phi^j\):  the function \(G_1(x) = \sum_j d_j \Phi^j(x)\) is not identically zero and has at most \(n\) zeros in \(\Pp(\F)\).  The set of \(x_0 \in \F^*\) with \(G_1(x_0) = 0\) therefore has size \(\le n\).

There are \(p^M - 1\) nonzero coefficient vectors in \(\F_p^M\).  By the union bound over all such vectors, plus the \(n\) structural exclusions:
\[
|\Bad| \le n + (p^M - 1) \cdot n \le n + n \cdot p^M.
\]
A good \(x_0\) exists when \(|\F^*| = p^d - 1 > n + n \cdot p^M\), which holds when \(p^d > 2n \cdot p^M\), i.e., \(M < d - \log_p(2n)\).
\end{proof}

\begin{theorem}[Universal multi-block]\label{thm:multi-block-universal}
Let \(p\) be any prime, \(d \ge 2\), and \(2 \le k \le n \le d(k-1)\) with \(n \le p^d - 1\), and let \(n \le M \le d(k-1)\).  To achieve perfect security against all admissible multi-block patterns with \(\le M\) blocks simultaneously, we get \(|\Bad| \le n + n \cdot (dpe)^M\).  A good \(x_0\) exists when \(d > M\big(\frac{5}{2} + \log_p d\big) + \log_p(2n)\).
\end{theorem}
\begin{proof}
We take a union bound over all admissible multi-block patterns with \(\le M\) blocks total.  An admissible pattern is specified by choosing, for each share \(j\), a nonempty subset of \(\{0, \ldots, d-1\}\) of block positions.  The total number of such patterns is at most \(\binom{d}{m_0} \cdots \binom{d}{m_{n-1}}\) summed over all \((m_0, \ldots, m_{n-1})\) with \(\sum m_j \le M\).  Also note that by definition of admissible patterns, we have \(m_0,\ldots,m_{n-1}\ge1\) and \(M\ge n\).  This gives us the total number as
\[
\sum_{\substack{m_0+\cdots+m_{n-1}\le M\\m_0,\ldots,m_{n-1}\ge1}}\binom{d}{m_0}\cdots\binom{d}{m_{n-1}}\le\sum_{t=n}^M\binom{nd}{t}\le\bigg(\frac{nde}{M}\bigg)^M\le(de)^M.
\]
where the first inequality is by the Chu-Vandermonde identity~\cite[Exercise 1.9]{jukna2011extremal}, the second inequality is the standard estimate on partial binomial sums~\cite[Proposition 1.4]{jukna2011extremal} (valid since \(1 \le M \le nd\)), and the third inequality is due to \(M \ge n\).

For each fixed pattern, \Cref{thm:multi-block-fixed-pattern} gives at most \(n \cdot p^M\) security-bad base points (beyond the structural exclusions).  Taking the union bound over all \(\le (de)^M\) patterns, we get
\[
|\Bad| \le n + (n \cdot p^M) \cdot (de)^M = n + n(dpe)^M.
\]
A good \(x_0\) exists when \(p^d - 1 > n + n(dpe)^M\).  Taking logarithms base \(p\),  we need \(d > \log_p(2n) + M(\log_p d + 1 + \log_p e)\), which is satisfied when \(d > M\big(\frac{5}{2} + \log_p d\big) + \log_p(2n)\).
\end{proof}

\cleardoublepage

%% ============================================================
%%  APPENDIX A
%% ============================================================
\appendix

\section{Alternative proof of the dichotomy}\label{app:dichotomy-proofs}

The perfect-or-completely-insecure dichotomy used throughout this paper (\Cref{pro:dichotomy}) is a special case of the general dichotomy result of~\cite{nguyen-2025}.  For completeness, we record here a direct derivation in our setting, working through the leakage map \(L_{\vec{i}}\) and the cosets-of-subspaces structure.  The proof depends only on \(\F_p\)-linearity of coordinate extraction and on Shamir's \(\F\)-linearity in the randomness;  it does not use any property of generalized Reed--Solomon codes or of binary images.

\begin{proof}[Proof of \Cref{pro:leakage-map}]
Since \(P(x_j) = s + \sum_{t=1}^{k-1} a_t x_j^t\), the \(\F_p\)-linearity of the coordinate map gives
\[
\ell_j = (P(x_j))_{i_j} = (s)_{i_j} + \sum_{t=1}^{k-1} (a_t x_j^t)_{i_j}.
\]
By \Cref{fac:coord-trace}, we have \((a_t x_j^t)_{i_j} = \Tr(a_t x_j^t \cdot f_{i_j}^*)\).  This gives the decomposition \(\ell_j = (s)_{i_j} + \sum_{t=1}^{k-1} \Tr(a_t \cdot x_j^t \cdot f_{i_j}^*)\).

It remains to verify that \(L_{\vec{i}}\) is \(\F_p\)-linear.  For any \(\lambda \in \F_p\) and \(\vec{a}, \vec{b} \in \F^{k-1}\):
\[
L_{\vec{i}}(\lambda \vec{a} + \vec{b})_j = \sum_{t=1}^{k-1} \Tr((\lambda a_t + b_t) \cdot x_j^t \cdot f_{i_j}^*) = \lambda \sum_{t=1}^{k-1} \Tr(a_t x_j^t f_{i_j}^*) + \sum_{t=1}^{k-1} \Tr(b_t x_j^t f_{i_j}^*),
\]
where we used \(\Tr(\lambda \cdot y) = \lambda \cdot \Tr(y)\) for \(\lambda \in \F_p\) (since the trace is \(\F_p\)-linear) and \(\Tr(y + z) = \Tr(y) + \Tr(z)\).  Thus \(L_{\vec{i}}(\lambda \vec{a} + \vec{b}) = \lambda L_{\vec{i}}(\vec{a}) + L_{\vec{i}}(\vec{b})\).
\end{proof}

\begin{proof}[Proof of \Cref{pro:dichotomy}]
Since \(\vec{a} = (a_1, \ldots, a_{k-1})\) is uniform over \(\F^{k-1}\) and \(L_{\vec{i}}\) is \(\F_p\)-linear, the random variable \(L_{\vec{i}}(\vec{a})\) is uniform over the subspace \(\Img(L_{\vec{i}}) \subseteq \F_p^n\).  By \Cref{pro:leakage-map}, the leakage for secret \(s\) is \(\vec{v}(s) + L_{\vec{i}}(\vec{a})\), which is uniform over the coset \(\vec{v}(s) + \Img(L_{\vec{i}})\).  Two cosets of an \(\F_p\)-subspace of \(\F_p^n\) are either identical (when the difference of the shifts lies in the subspace, giving \(\SD = 0\)) or disjoint (giving \(\SD = 1\), since uniform distributions on disjoint sets of equal size have statistical distance \(1\)).  If \(L_{\vec{i}}\) is surjective, then \(\Img(L_{\vec{i}}) = \F_p^n\), so every coset equals \(\F_p^n\).
\end{proof}

\cleardoublepage

% ============================================================
% BIBLIOGRAPHY
% ============================================================
\ifanonymous
\printbibliography
\else
% Use adjustwidth instead of newgeometry to prevent the page break.
% Since your default margin is 2.5cm and you want 1.8cm, 
% we pull the left and right margins out by -0.7cm.
\begin{adjustwidth}{-0.7cm}{-0.7cm}
	\raggedright
	\begin{multicols}{2}[\printbibheading]
		\printbibliography[heading=none]
	\end{multicols}
\end{adjustwidth}
\fi

\cleardoublepage
\end{document}